\documentclass[conference,a4paper]{IEEEtran}

\usepackage[cmex10]{amsmath}
\usepackage{graphicx}
\usepackage{latexsym}
\usepackage{amssymb}
\usepackage{bm}
\usepackage{psfrag}

\newtheorem{definition}{Definition}
\newtheorem{theorem}{Theorem}
\newtheorem{lemma}{Lemma}
\newtheorem{proposition}{Proposition}
\newtheorem{corollary}{Corollary}
\newtheorem{remark}{Remark}

\newcommand{\bs}{\boldsymbol}
\newcommand{\coef}{\mathrm{coef}}
\newcommand{\diff}[1]{d#1}

\newcommand{\deri}[2]{\frac{\diff{#1}}{\diff{#2}}}
\newcommand{\neib}{\mathcal{N}}
\newcommand{\ens}{\mathrm{E}}

\newcommand{\mat}[1]{{\mathbf #1}}
\newcommand{\diag}{\mathrm{diag}}
\newcommand{\expect}{\mathbb{E}}
\newcommand{\Int}[1]{[\![#1]\!] }

\begin{document}

\sloppy

\title{
Cutsize Distributions of Balanced Hypergraph Bipartitions for Random Hypergraphs
}

\author{
\IEEEauthorblockN{Takayuki Nozaki}
\IEEEauthorblockA{
Yamaguchi University, JAPAN\\
Email: tnozaki@yamaguchi-u.ac.jp
}
}
\maketitle

\begin{abstract}
In a previous work, we presented a parallel encoding algorithm for low-density parity-check (LDPC) codes by partitioning hypergraph representation for the LDPC codes.
The aim of this research is to analyze the processing time of this encoding algorithm.
This paper clarifies that the processing time of the encoding algorithm depends on the minimum cutsize of balanced hypergraph partitions.
Moreover, this paper gives the typical minimum cutsize and cutsize distribution for balanced hypergraph bipartitions of random hypergraphs defined from a regular LDPC ensemble.
\end{abstract}

\section{Introduction}
Low-density parity-check (LDPC) code \cite{Gallager_LDPC} is a linear code defined by a sparse parity check matrix $\mathbf{H}$.
The encoding algorithms generate the codeword ${\bs x} = ({\bs p}, {\bs m})$ from a give message ${\bs m}$.
Since ${\bs 0} = \mathbf{H}{\bs x}^T = \begin{pmatrix} \mathbf{H}_P & \mathbf{H}_I \end{pmatrix} ({\bs p}, {\bs m})^T$, 
the encoding algorithm solves a system of linear equations
\begin{equation}
 \mathbf{H}_P {\bs p}^T = - \mathbf{H}_I {\bs m}^T. \label{eq:enc}
\end{equation}
Hence, if we can transform the given parity check matrix to a matrix $\begin{pmatrix} \mathbf{H}_P & \mathbf{H}_I \end{pmatrix}$ such that \eqref{eq:enc} is efficiently solved, we can obtain an efficient encoding algorithm.

We presented an efficient parallel encoding algorithm for LDPC codes by transforming $\mathbf{H}_P$ into a block diagonal matrix $\diag[\mathbf{H}_{P,1}, \mathbf{H}_{P,2},\dots, \mathbf{H}_{P,K}]$ \cite{nozaki2015}.
More precisely, this algorithm breaks down the system of linear equations \eqref{eq:enc} into $K$ systems of linear equations and parallelly solves those $K$ systems of linear equations.
We showed in \cite{nozaki2015} that the total number of operations of this encoding algorithm approximately equals to Richardson and Urbanke's (RU) encoding algorithm \cite{RU_eff}.
Since this encoding algorithm simultaneously solves $K$ systems of linear equations,
the processing time of the encoding algorithm is $1/K$ of the RU encoding algorithm.

The aim of this research is to analyze the processing time of this encoding algorithm.
In other words, we would like to analyze the parallel degree $K$, where the maximum of $K$ depends on the given parity check matrix.

This paper clarifies that $K$ depends on the minimum cutsize in balanced partitions for the hypergraph representation to $\mathbf{H}$.
However, it is known that the balanced hypergraph partitioning problem, which divides the vertices of a hypergraph into $K$ almost equal size parts, is NP-hard \cite{lengauer1990}.
Hence, it is difficult to calculate the minimum cutsize for a given hypergraph representation to $\mathbf{H}$.

In this paper, we take a coding theoretic approach to evaluate the minimum cutsize: (1) considering a random hypergraph ensemble, (2) deriving the ensemble average of balanced partitions with a given cutsize, i.e, deriving the cutsize distribution, (3) analyzing the growth rate for the cutsize distribution, and (4) clarifying the typical minimum cutsize for the hypergraph ensemble.
In other words, we use a similar technique to derive minimum distance for the LDPC ensembles \cite{1302293}.
In this paper, we derive the typical minimum cutsize of balanced bipartitions, i.e, $K=2$, for random hypergraph ensemble defined from regular LDPC ensemble, as a first step of the research.


As related works, Wadayama et al.\ \cite{fujii2012coding, yano2012probabilistic,fujii2013analysis} analyzed random graphs by using coding theoretic approaches.
Dembo et al.\ \cite{dembo2015extremal} evaluated cutsize in random {\it graph} bisections.

The remainder of the paper is organized as follows.
Section \ref{sec:pre} gives hypergraph representation for LDPC codes and balanced hypergraph partitioning, and introduces regular LDPC ensembles and corresponding random hypergraph ensemble.
Section \ref{sec:enc} derives a necessary condition for $K$-parallel encodable LDPC codes.
In other words, we will show that $K$ depends on the minimum cutsize in hypergraph partitions in Section \ref{sec:enc}.
Section \ref{sec:cutdist} gives the cutsize distribution of balanced hypergraph bipartitions for the hypergraph ensembles.
Section \ref{sec:tmimcut} analyzes the groth rate of the cutsize distribution and typical minimum cutsize of balanced hypergraph bipartitions for the hypergraph ensembles.
Section \ref{sec:conc} concludes the paper.

\section{Preliminaries \label{sec:pre}}
This section introduces LDPC codes, hypergraph partitioning and random hypergraph ensemble.

\subsection{Hypergraph Representation for LDPC code \label{ssec:hrep}}
This section introduces three representations for LDPC codes, namely, parity check matrix, Tanner graph and hypergraph, and gives the relationship between Tanner graph and hypergraph \cite{nozaki2015}.

The Tanner graph $\mathtt{G} = (\mathtt{V}\cup \mathtt{C}, \mathtt{E})$ for a LDPC code is represented by a bipartite graph with the set of variable nodes $\mathtt{V}$, check nodes $\mathtt{C}$ and edges $\mathtt{E}$.
For a given $m\times n$ parity check matrix $\mat{H} = (h_{i,j})$,
the $j$-th variable node $\mathtt{v}_j$ and $i$-th check node $\mathtt{c}_i$ are connected iff $h_{i,j}\neq 0$, i.e, $(\mathtt{v}_j,\mathtt{c}_i)\in\mathtt{E} \iff h_{i,j} \neq 0$.
In other words, the $i$-th check node (resp.\ $j$-th variable node) in Tanner graph $\mathtt{G}$ corresponds to the $i$-th row (resp.\ $j$-th column) of parity check matrix $\mat{H}$.

Let $\mathcal{U}$ be a finite set, and let $\mathcal{E}$ be a family of non-empty subsets of $\mathcal{U}$.
The pair $\mathcal{H} = (\mathcal{U}, \mathcal{E})$ is called {\it hypergraph} with the set of vertices $\mathcal{U}$ and the set of {\it nets} (or {\it hyperedges}) $\mathcal{E}$.
If the $i$-th node $u_i\in\mathcal{U}$ is in the $j$-th net $e_j\in\mathcal{E}$, i.e, $u_i\in e_j$, the vertex $u_i$ is {\it connected} to $e_j$.
For a given $m\times n$ matrix $\mat{H}=(h_{i,j})$, the hypergraph representation $\mathcal{H}_{\mat{H}} = (\mathcal{U},\mathcal{E})$ is constructed in the following way.
The number of vertices $|\mathcal{U}|$ is $m$ and the number of nets $|\mathcal{E}|$ is $n$.
The vertex $u_i$ is connected to the net $e_j$ iff $h_{i,j}\neq 0$, i.e, $u_i\in e_j \iff h_{i,j}\neq 0$.
In other words, the $i$-th vertex (resp.\ $j$-th net) corresponds to the $i$-th row (resp.\ $j$-th column).

By summarizing above, if we transform the variable nodes (resp.\ check nodes) in $\mathtt{G}$ to nets (resp.\ vertices), we can obtain the hypergraph representation $\mathcal{H}$ for the LDPC code defined by the Tanner graph $\mathtt{G}$.

\subsection{Balanced Hypergraph Partitioning \cite{780863}}
A family $\Pi_K = \{\mathcal{U}_1, \mathcal{U}_2, \dots, \mathcal{U}_K\}$ of non-empty subsets of $\mathcal{U}$ is a {\it $K$-way hypergraph partition} of $\mathcal{H}=(\mathcal{U},\mathcal{E})$ if the followings are satisfied:
\begin{itemize}
\item 
Each pair of parts is disjoint, i.e, $\mathcal{U}_i\cap \mathcal{U}_j = \emptyset$ for all $1\le i < j \le K$.
\item
Union of $K$ parts is equal to $\mathcal{U}$, i.e, $\bigcup_{i=1}^{K} \mathcal{U}_i = \mathcal{U}$.
\end{itemize}
In particular, two-way hypergraph partition is called {\it hypergraph bipartition}.

For a fixed partition $\Pi_K$, if a net $e\in\mathcal{E}$ connects to a node $u$ in a part $\mathcal{U}_i$, we call that the net $e$ connects to the part $\mathcal{U}_i$.
Denote the set of nets connecting to a part $\mathcal{U}_i$, by $\mathcal{N}(\mathcal{U}_i)$.
A net is called {\it cut} if the net connects to more than one parts.
For a fixed partition $\Pi_K$, the set of cuts is called {\it cut set} and denoted by $\mathcal{X}(\Pi_K)$.
The {\it cutsize} of $\Pi_K$ is given by $|\mathcal{X}(\Pi_K)|$.

We denote the set of integers between $a$ and $b$, by
\begin{equation*}
 \Int{a,b} := \{k\in \mathbb{Z} \mid a\le k \le b\}.
\end{equation*}
A partition is $\epsilon$-{\it balanced} \cite{780863} if the following holds:
\begin{equation*}
 \max_{i\in\Int{1,K}} |\mathcal{U}_i| \le \frac{|\mathcal{U}|}{K}(1+\epsilon),
\end{equation*}
where $\epsilon \ge 0$ represents the predetermined maximum imbalance ratio.
We denote a $K$-way $\epsilon$-balanced partition, by $\Pi_{K}^{(\epsilon)}$.
In particular, a partition is {\it exactly balanced} if $\epsilon = 0$.

\subsection{Regular LDPC Ensemble and Random Hypergraph}
For a given $n, \gamma, \delta$, an LDPC ensemble $\ens(n,\gamma,\delta)$ is defined by the following way.
There exist $n$ variable nodes of degree $\gamma$ and 
$m$ check nodes of degree $\delta$.
A node of degree $i$ has $i$ sockets for its connected edges.
Consider a permutation $\pi$ on the number of edges $\xi := \gamma n$.
Join the $i$-th socket on the variable node side to the $\pi(i)$-th socket
on the check node side.
The bipartite graphs are chosen with equal probability
from all the permutations on the number of edges.

From Section \ref{ssec:hrep}, we can generate a hypergraph $\mathcal{H}$ from a Tanner graph $\mathtt{G}$. 
Hence, an LDPC ensemble is regarded as a hypergraph ensemble.
With some abuse of notation, we denote $\mathcal{H}\in\ens(n,\gamma,\delta)$ if the corresponding Tanner graph $\mathtt{G}$ belongs to $\ens(n,\gamma,\delta)$.

\section{Condition for Parallel Encodable \label{sec:enc}}
In this section, we will derive a necessary condition that an LDPC code is $K$ parallel encodable.

\begin{definition}[$K$ parallel encodable] \label{def:penc}
Assume that $K$ integers $m_1,m_2,\dots, m_K$ satisfies $\sum_{i} m_i = m$ and $\max_i m_i \le (1+\epsilon)m/K$.
For a given parity check matrix $\mat{H}$, 
an LDPC code is $K$ {\it parallel encodable by block-diagonalization} if
there exists a pair of permutation matrices $\mat{P}, \mat{Q}$ such that
\begin{equation}
 \mathbf{P}\mathbf{H}\mathbf{Q}
  =
 \begin{pmatrix}
  \mathbf{H}_P & \mathbf{H}_I
 \end{pmatrix}
  =
 \begin{pmatrix}
  \mathbf{H}_{P,1} & & \mathbf{O} & \mathbf{H}_{I,1}\\
              &\ddots &  & \vdots \\
  \mathbf{O} & & \mathbf{H}_{P,K} & \mathbf{H}_{I,K}
 \end{pmatrix}, \label{eq:penc}
\end{equation}
and $\mat{H}_{P,i}$ is a non-singular $m_i \times m_i$ matrix for $i=1,2,\dots, K$.
\end{definition}

Split the parity part ${\bs p}$ into $K$ parts $({\bs p}_1, {\bs p}_2,\dots, {\bs p}_K)$, where the length of ${\bs p}_i$ is $m_i$.
Then, we obtain the parity part ${\bs p} = ({\bs p}_1, {\bs p}_2,\dots, {\bs p}_K)$ by parallelly solving the systems of linear equations $\mathbf{H}_{P,i} {\bs p}_i^T = -\mathbf{H}_{I,i} {\bs m}^T$ if the LDPC code $\mat{H}$ is $K$ parallel encodable by block-diagonalization.
Note that the systems of linear equations are almost equal size since $\max_i m_i \le (1+\epsilon)m/K$.

The following proposition gives a necessary condition that $\mat{H}$ is $K$ parallel encodable by block-diagonalization.

\begin{proposition} \label{prop:penc}
If an LDPC code defined by $\mat{H}$ is $K$ parallel encodable by block-diagonalization, the following condition holds:
\begin{equation}
 n-m
  \ge
 \min_{\Pi_K^{(\epsilon)}} |\mathcal{X}(\Pi_K^{(\epsilon)})|. \label{eq:para_cond}
\end{equation}
\end{proposition}
\begin{proof}
From Definition \ref{def:penc}, if an LDPC code defined by $\mat{H}$ is $K$ parallel encodable by block-diagonalization, a pair of permutation matrices $\mat{P},\mat{Q}$ transforms $\mat{H}$ as \eqref{eq:penc}.
Denote the set of vertices corresponding to $\Int{\sum_{j=1}^{i-1}m_j+1, \sum_{j=1}^i m_j}$ rows of $\mat{P}\mat{H}\mat{Q}$, by $\mathcal{U}_i$.
Since $|\mathcal{U}_i| = m_i$ and $m_i \le (1+\epsilon)m/K$, 
the partition $(\mathcal{U}_1,\mathcal{U}_2,\dots, \mathcal{U}_K)$ is $\epsilon$-balanced.
Let $\mathcal{E}_i$ be the set of nets corresponding to $\Int{\sum_{j=1}^{i-1}m_j+1, \sum_{j=1}^i m_j}$ columns of $\mat{P}\mat{H}\mat{Q}$.
Since the elements in $\mathcal{E}_i$ only connect to $\mathcal{U}_i$,
$\mathcal{E}_i \subseteq \neib(\mathcal{U}_i) \setminus \mathcal{X}(\Pi_K^{(\epsilon)})$ holds.
Noticing that $|\mathcal{E}_i| = m_i$, we have
\begin{equation}
 m_i \le |\mathcal{N}(\mathcal{U}_i)\setminus \mathcal{X}(\Pi_K^{(\epsilon)})|. \notag
\end{equation}
By summing up this equation over $i$, we get
\begin{equation}
 m 
  \le
 \sum_{i=1}^{K}|\mathcal{N}(\mathcal{U}_i)\setminus \mathcal{X}(\Pi_K^{(\epsilon)})| 
  =
 n - |\mathcal{X}(\Pi_K^{(\epsilon)})| \notag
\end{equation}
From this inequation, we obtain \eqref{eq:para_cond}.
\end{proof}

For a fixed $\mat{H}$, $\min_{\Pi_K^{(\epsilon)}} |\mathcal{X}(\Pi_K^{(\epsilon)})|$ does not decrease as $K$ increases.
Hence, there exists the maximum parallel degree $K_{\max} := \max\{K \mid  n-m \ge \min_{\Pi_K^{(\epsilon)}} |\mathcal{X}(\Pi_K^{(\epsilon)})|\}$.
Thus, to analyze the processing time for the parallel encoding algorithm, we need to calculate $\min_{\Pi_K^{(\epsilon)}} |\mathcal{X}(\Pi_K^{(\epsilon)})|$ for a given $\mat{H}$.

However, it is known that the balanced hypergraph partitioning problem is NP-hard \cite{lengauer1990}. 
In other words, it is difficult to calculate the minimum cutsize $\min_{\Pi_K^{(\epsilon)}} |\mathcal{X}(\Pi_K^{(\epsilon)})|$ for a given $\mat{H}$.
Hence, we will analyze the typical minimum cutsize of hypergraph bipartitioning, i.e, $K=2$, for a fixed ensemble $\ens(n,\gamma,\delta)$ in the following sections.

\section{Cutsize Distribution \label{sec:cutdist}}
In this section, we derive cutsize distribution of hypergraph bipartitioning for $\ens(n,\gamma,\delta)$.

\begin{definition}[Cutsize distirbution]
For a hypergraph $\mathcal{H}$, let $A_{\mathcal{H}}(s, m_1)$ be the number of 
bipartitions such that $|\mathcal{X}(\Pi_2)| = s$ and $|\mathcal{U}_1| = m_1$.
For an ensemble $\ens(n,\gamma,\delta)$,
the cutsize distribution $A(s,m_1)$ is the ensemble average of $A_{\mathcal{H}}(s, m_1)$, i.e, 
\begin{align}
 A(s,m_1) 
 &:= \expect_{\mathcal{H}\in\ens(n,\gamma,\delta)}[A_{\mathcal{H}}(s, m_1)]  \notag \\
 &= \frac{1}{\xi !} \sum_{\mathcal{H} \in\ens(n,\gamma,\delta)} A_{\mathcal{H}}(s, m_1).
\end{align}
Similarly, for a hypergraph $\mathcal{H}$, let $B_{\mathcal{H}}(s, \epsilon)$ be the number of $\epsilon$-balanced bipartitions with cutsize $s$.
For an ensemble $\ens(n,\gamma,\delta)$, the cutsize distribution $B(s,\epsilon)$ is defined by the ensemble average of $B_{\mathcal{H}}(s, m_1)$.
\end{definition}

Since the partitions are $\epsilon$-balanced, $|\mathcal{U}_1| = m_1 \le m(1+\epsilon)/2$ and $|\mathcal{U}_2| = m - m_1 \le m(1+\epsilon)/2$ hold.
Hence, $m_1 \in M_\epsilon := \Int{m(1-\epsilon)/2, m(1+\epsilon)/2}$.
Then, $B(s,\epsilon)$ is given by $A(s,m_1)$ as follows:
\begin{equation}
B(s,\epsilon) = \sum_{m_1\in M_{\epsilon}}A(s,m_1). \label{eq:BbyA}
\end{equation}

The following theorem gives the cutsize distribution $A(s,m_1)$ for $\ens(n,\gamma,\delta)$.
\begin{theorem} \label{the:cut_wd}
 For an ensemble $\ens(n,\gamma,\delta)$, the cutsize distribution $A(s,m_1)$ is given as follows:
 \begin{align}
  &A(s,m_1) 
   = 
  \frac{\binom{m}{m_1}  \binom{n}{s}}{\binom{\delta m}{\delta m_1}}
  \coef ( f(u)^n ,u^{\delta m_1}) \notag \\
  &\qquad\qquad\qquad
  \times \mathbb{I}[s \le \delta m_1]\mathbb{I}[s \le \delta (m-m_1)], 
  \label{eq:cut_wd} \\
  &f(u) := p(u)^{s/n} q(u)^{1-s/n}, \\
  &p(u) := (1+u)^{\gamma} - 1 -u^{\gamma},\quad 
  q(u) := 1 + u^{\gamma}. \label{eq:pq}
 \end{align}
 where $\coef(f(x),x^i)$ is the coefficient of $x^i$ in the polynomial $f(x)$ and 
\begin{equation*}
 \mathbb{I}[P]  = 
  \begin{cases}
   1 & \text{if $P$ is true}, \\
   0 & \text{otherwise}.
  \end{cases}
\end{equation*}
\end{theorem}
\begin{proof}
Notice that $|\mathcal{U}_2| = m-m_1$.
The number of nets connecting to $\mathcal{U}_1$ (resp.\ $\mathcal{U}_2$) is at most $\delta m_1$ (resp.\ $\delta (m-m_1)$), i.e, $|\neib (\mathcal{U}_1)| \le \delta m_1$ (resp.\ $|\neib (\mathcal{U}_2)| \le \delta (m-m_1)$).
The cutsize $s$ is smaller than $|\neib (\mathcal{U}_1)|$ and $|\neib (\mathcal{U}_2)|$.
Hence $A(s,m_1) = 0$ if $s>\delta m_1$ or $s>\delta (m-m_1)$.

Fix $s$ and $m_1$ with $s \le \delta m_1$ and $s \le \delta (m - m_1)$.
For a fixed $\mathcal{H}\in\ens(n,\gamma,\delta)$ and $\Pi_2 = (\mathcal{U}_1, \mathcal{U}_2)$, let $\mathtt{C}_1$ (resp.\ $\mathtt{C}_2$) be the set of the check nodes corresponding to $\mathcal{U}_1$ (resp.\ $\mathcal{U}_2$) and let $\mathtt{V}_{\mathrm{i}}$ be the set of the variable nodes corresponding to the cut set $\mathcal{X}(\Pi_2)$.
The total number of Tanner graphs such that $|\mathtt{C}_1| = m_1$ and $|\mathtt{V}_{\mathrm{i}}| = s$ equals to $A(s, m_1) \xi!$.
We refer the edges connecting to $\mathtt{C}_1$ as {\it active edges}.
The socket is active if the connecting edge is active.
Since $|\mathtt{C}_1| = m_1$, the number of active edges is $\delta m_1$.

We count the number of constellations for active sockets in the variable node side.
For a variable node $\mathtt{v}$ of degree $\gamma$, let $a_{i,j}$ be the number of constellations for $j$ active sockets in $\mathtt{v}$, where $i=1$ if $\mathtt{v}$ belongs to $\mathtt{V}_{\mathrm{i}}$, otherwise $i=0$.
Note that the variable node $\mathtt{v}$ belongs to  $\mathtt{V}_{\mathrm{i}}$ iff $\mathtt{v}$ has $j\in\Int{1,\gamma-1}$ active sockets.
Hence, we get
\begin{align*}
 &a_{0,j} = \binom{\gamma}{j} \mathbb{I}[1\le j \le \gamma-1], \\
 &a_{1,j} = \mathbb{I}[j=0 \text{~or~} j=\gamma].
\end{align*}
The generating function $a(t,u) := \sum_{i,j} a_{i,j}t^i u^j$ is given as 
\begin{equation}
 a(t,u) 
 = t\{(1+u)^{\gamma}-1-u^{\gamma}\} + (1+u^{\gamma}) 
 =: tp(u) + q(u). \notag
\end{equation}
Thus, the number of constellations of $\delta m_1$ active sockets in variable node side for a given cutsize $s$ is 
\begin{align}
 &\coef(a(t,u)^n, t^s u^{\delta m_1}) \notag \\ 
  &\quad=
 \binom{n}{s} \coef( \{p(u)\}^s\{q(u)\}^{n-s}, u^{\delta m_1})  \notag \\
  &\quad=:
 \binom{n}{s} \coef( f(u)^n, u^{\delta m_1} ). \label{eq:vcons}
\end{align}

The number of choices for the check nodes in $\mathtt{C}_1$ is $\binom{m}{m_1}$.
The number of permutations of active edges (resp.\ non-active edges) is $(\delta m_1)!$ (resp.\ $(\delta m-\delta m_1)!$).
By multiplying those numbers and \eqref{eq:vcons}, we have
\begin{align*}
 A(s,m_1) \xi ! 
  =
 \textstyle{\binom{m}{m_1} (\delta m_1)! (\delta m-\delta m_1)! 
 \binom{n}{s} \coef( f(u)^n, u^{\delta m_1} )},
\end{align*}
for $s \le \delta m_1$ and $s \le \delta (m - m_1)$.
From this equation, we get \eqref{eq:cut_wd}.
\end{proof}

\section{Typical Minimum Cutsize \label{sec:tmimcut}}
In this section, we analyze the asymptotic behavior of the cutsize distributions, namely, the growth rate for the cutsize distributions and the relative typical minimum cutsizes for the ensemble $\ens(n,\gamma,\delta)$.
Firstly, we define the growth rate and the relative typical minimum cutsizes, and explain the meanings of those terms.

\begin{definition}[Growth rate]
Consider $\ens(n,\gamma,\delta)$.
Define the growth rate $g(\sigma, \mu_1)$ and $h(\sigma, \epsilon)$ for the cutsize distributions $A(\sigma n, \mu_1 m)$ and $B(\sigma n, \epsilon)$ as
\begin{align*}
 &g(\sigma, \mu_1) 
  = 
 \lim_{n\to\infty} \frac{1}{n} \log A(\sigma n, \mu_1 m), \\
 &h(\sigma, \epsilon)
  =
 \lim_{n\to\infty} \frac{1}{n} \log B(\sigma n, \epsilon),
\end{align*}
respectively. 
\end{definition}

The expression of the growth rates are given in Section \ref{ssec:gr}.

\begin{remark} \label{rem:1}
From the definition of growth rates, $A(\sigma n,\mu_1 m) \approx 2^{n g(\sigma, \mu_1)}$ and $B(\sigma n, \epsilon) \approx 2^{n h(\sigma, \epsilon)}$ hold.
This implies that $A(\sigma n, \mu_1 m)$ (resp.\ $B(\sigma n, \epsilon)$) is exponentially decreasing for $n$ if $g(\sigma, \mu_1) < 0$ (resp.\ $h(\sigma, \epsilon) < 0$).
Similarly, $A(\sigma n, \mu_1 m)$ (resp.\ $B(\sigma n, \epsilon)$) is exponentially increasing for $n$ if $g(\sigma, \mu_1) > 0$ (resp.\ $h(\sigma, \epsilon) > 0$). 
\end{remark}

\begin{definition}[Relative typical minimum cutsize]
Consider $\ens(n,\gamma,\delta)$.
 Define 
 \begin{align}
 &\alpha^*(\mu_1) := \inf \{\sigma>0 \mid g(\sigma,\mu_1)>0\}, \\
 &\beta^*(\epsilon) := \inf \{\sigma > 0 \mid h(\sigma, \epsilon) > 0\}.
 \end{align}
 We refer the value $\alpha^*(\mu_1)$ and $\beta^*(\epsilon)$ as 
 the {\it relative typical minimum cutsizes} for $\ens(n,\gamma,\delta)$.
\end{definition}

We will show the existence of the relative typical minimum cutsizes in Section \ref{sssec:etmc}.

As discussed in Remark \ref{rem:1}, the number of $\epsilon$-balanced partitions with cutsize $\sigma n$ exponentially decreases as $n\to\infty$ for $\sigma \in [0,\beta^*(\epsilon))$.
In other words, 
the minimum cutsize $\min_{\Pi_2^{(\epsilon)}} |\mathcal{X}(\Pi_2^{(\epsilon)})|$
 is approximated by $\beta^*(\epsilon) n$ for almost all $\mathcal{H} \in \ens(n,\gamma,\delta)$.
Hence, we refer $\beta^*(\epsilon)$ as the typical minimum cutsize.

From the above discussion and Proposition \ref{prop:penc}, we obtain a necessary condition that $\mat{H}\in\ens(n,\gamma,\delta)$ is $K=2$ parallel encodable by the block-diagonalization with high probability.
\begin{proposition} \label{prop:penc_ens}
 If a code $\mat{H}\in\ens(n,\gamma,\delta)$ is $K=2$ parallel encodable by the block-diagonalization with high probability, the following condition holds:
\begin{equation}
 1-\frac{\delta}{\gamma}
  \ge
 \beta^{*}(\epsilon). \label{eq:para_cond2}
\end{equation}
\end{proposition}

Thus, the typical minimum cutsize is an important characteristic for parallel encodable by the block-diagonalization.
We will evaluate both the right and left hand sides of \eqref{eq:para_cond2} in Section \ref{ssec:tmc} for some ensembles by numerical examples.

\subsection{Growth Rate \label{ssec:gr}}

From \eqref{eq:BbyA}, we have the following relationship between $g(\sigma,\mu_1)$ and $h(\sigma, \epsilon)$.
\begin{lemma} \label{lem:gh}
 Define $\bar{M}_{\epsilon} := [(1-\epsilon)/2, (1+\epsilon)/2]$.
 Then, the following equation holds for all $\sigma \in [0,1]$:
\begin{equation}
 h(\sigma,\epsilon) 
 = 
 \max_{\mu_1\in \bar{M}_{\epsilon}} g(\sigma, \mu_1). \label{eq:growth_h_rep}
\end{equation}
\end{lemma}

From \cite[Theorem 2]{1302293}, the following equation holds:
\begin{align*}
 \lim_{n\to\infty} \frac{1}{n}\log \coef (f(u), u^{\alpha n}) 
  = 
 \inf_{u>0} \log \frac{f(u)}{u^{\alpha}} 
\end{align*}
where a point $u$ achieving the infimum satisfies $u \deri{f}{u} = \alpha f(u)$.
Moreover, it is well known that the following equation holds:
\begin{align*}
 \lim_{n\to\infty} \frac{1}{n} \log \binom{\alpha n}{\beta n} 
  =  
 \alpha H_2 (\beta/\alpha),
\end{align*}
where $H_2(x)$ is the binary entropy function, i.e, 
$H_2 (x) = -x \log x -(1-x)\log (1-x)$．

The growth rate $g(\sigma,\mu_1)$ is derived from those equations and Theorem \ref{the:cut_wd} as follows.
\begin{theorem} \label{the:gr}
Assume $\sigma \le \gamma \mu_1$ and $\sigma \le \gamma (1-\mu_1)$.
Define $p(u)$ and $q(u)$ as in \eqref{eq:pq}.
The growth rate $g(\sigma, \mu_1)$ for the ensemble $\ens(n,\gamma,\delta)$ is given as
\begin{align} 
 &g(\sigma,\mu_1)
  =
  H_2(\sigma)
 -\gamma\frac{\delta-1}{\delta} H_2 (\mu_1) \notag \\
 &\qquad\qquad+\inf_{u>0} \{ \sigma \log p(u) + (1-\sigma) \log q(u) - \mu_1 \gamma \log u \}.
  \label{eq:gamma} 
\end{align}
A point $u$ achieving the infimum satisfies
\begin{align}
  \sigma  u p'(u) q(u) + (1-\sigma) u p(u) q'(u) 
  = \mu_1 \gamma p(u) q(u), \label{eq:cond_u2}
\end{align}
where $p'(u) := \deri{p}{u}$.
\end{theorem}

Combining Lemma \ref{lem:gh} and Theorem \ref{the:gr}, the growth rate $h(\sigma,\epsilon)$ is expressed as
\begin{align}
 &h(\sigma,\epsilon) 
  =
   \max_{\mu_1 \in \bar{M}_{\epsilon}}\biggl[ H_2(\sigma)
 -\gamma\frac{\delta-1}{\delta} H_2 (\mu_1) \notag \\
 &\quad\qquad+\inf_{u>0} \{ \sigma \log p(u) + (1-\sigma) \log q(u) - \mu_1 \gamma \log u \} \biggr]. \label{eq:hgr}
\end{align}

\subsubsection{Existence of Typical Minimum Cutsize \label{sssec:etmc}}
In this section, we show that there exists a relative typical minimum cutsize $\alpha^*(\mu_1), \beta^{*}(\epsilon)$ for $\ens(n,\gamma,\delta)$.

\begin{lemma}\label{lem:minu}
 For $\ens(n,\gamma,\delta)$, the growth rates at $\sigma = 0$ are given as
\begin{align}
 &g(0,\mu_1)
  =
 \Bigl( 1- \gamma\frac{\delta-1}{\delta} \Bigr) H_2 (\mu_1),
 \label{eq:gamma_0} \\
 &h(0,\epsilon)
  =
 \Bigl(1-\gamma\frac{\delta-1}{\delta} \Bigr) H_2 \Bigl(\frac{1-\epsilon}{2}\Bigr).
 \label{eq:hgamma_0}
\end{align}
\end{lemma}
\begin{lemma} \label{lem:posi}
 Assume $\ens(n,\gamma,\delta)$. 
 For a fixed $\mu_1$, the maximum of growth rate $g(\sigma, \mu_1)$ is achieved at $\sigma^+ := 1-(1-\mu_1)^{\gamma}-\mu_1^{\gamma}$ and the maximum value is 
 \begin{equation}
  g(\sigma^+,\mu_1) = \frac{\gamma}{\delta}H_2(\mu_1) > 0.
 \end{equation}
\end{lemma}

Assume $\gamma \ge 2$ and $\delta \ge 3$.
Then, $g(0,\mu_1) \le 0$ and $h(0,\epsilon) \le 0$ hold from \eqref{eq:gamma_0} and \eqref{eq:hgamma_0}, respectively.
Note that the growth rates $g(\sigma, \mu_1)$ and $h(\sigma,\epsilon)$ are continuous functions for $\sigma$.
Hence, from  Lemmas \ref{lem:minu} and \ref{lem:posi}, the following proposition holds.
\begin{proposition} \label{prop:zero}
Assume $\ens (n,\gamma,\delta)$ with $\gamma\ge 2$ and $\delta \ge 3$.
For a fixed $\mu_1$, there exist $\sigma_0 \in [0,\sigma^+]$ such that $g(\sigma_0,\mu_1) = 0$.
Similarly, for a fixed $\epsilon$, there exist $\sigma_0 \in [0,\sigma^+]$ such that $h(\sigma_0,\epsilon) = 0$.
\end{proposition}

\subsubsection{A Closed Form Lower Bound}
In this section, we give a closed form lower bound for the growth rate $h(\sigma, \epsilon)$ given in \eqref{eq:hgr}.
Lemma \ref{lem:gh} gives a lower bound of $h(\sigma,\epsilon)$.
\begin{corollary} \label{cor:lower}
 The growth rate $h(\sigma,\epsilon)$ is lower bounded as
 \begin{equation}
  h(\sigma, \epsilon) > h(\sigma, 0) = g(\sigma, 1/2).
 \end{equation}
\end{corollary} 
The growth rate $g(\sigma, 1/2)$ has the following closed form expression.
\begin{lemma} \label{lem:g.5}
 For $\ens(n, \gamma,\delta)$, 
\begin{align} 
  g(\sigma, 1/2) 
   =
  H_2(\sigma) + \sigma \log (2^{\gamma-1}-1) -\gamma \frac{\delta-1}{\delta} + 1.
 \label{eq:gamma.5}
\end{align}
\end{lemma}

From Corollary \ref{cor:lower} and Lemma \ref{lem:g.5}, the growth rate $h(\sigma, \epsilon)$ is lower bounded by the left hand side of \eqref{eq:gamma.5}.
Moreover, Lemma \ref{lem:g.5} shows that the growth rate for the exactly balanced bipartitioning is written in a closed form.

\begin{figure}
 \begin{center}
 \includegraphics[width=.8\linewidth]{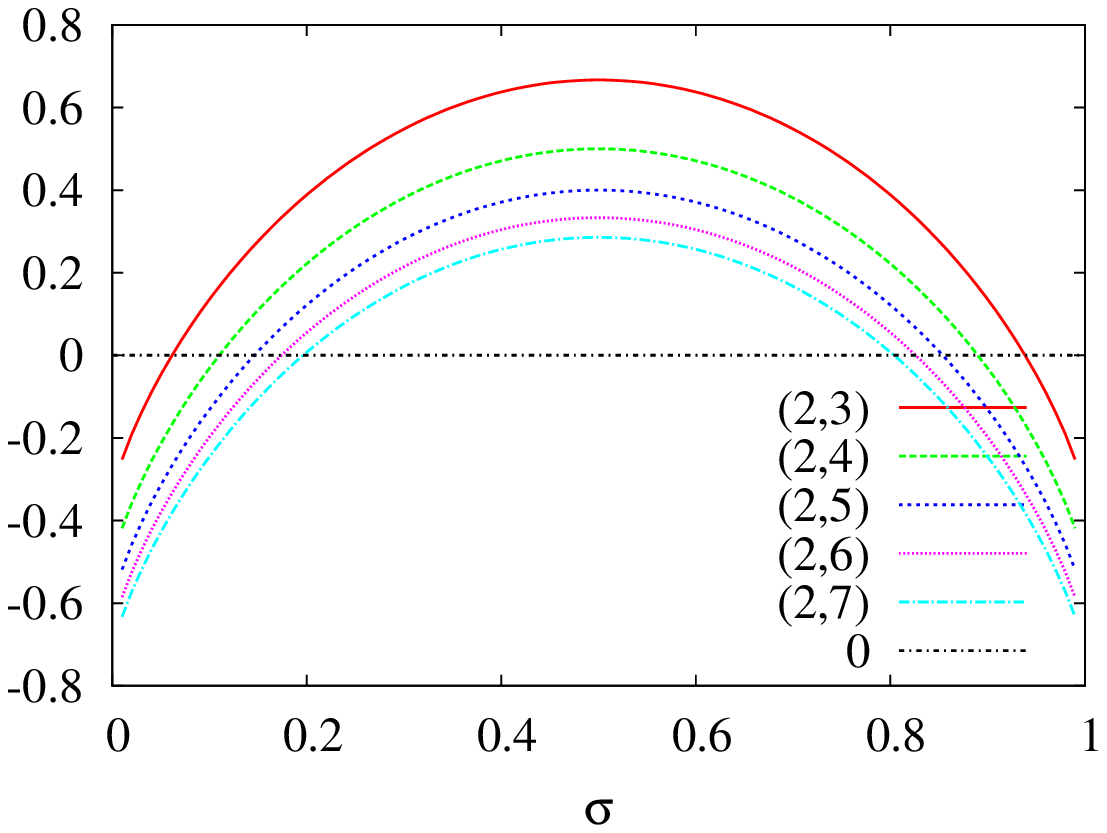}
 \caption{Growth rate $h(\sigma,0)$ for $\ens(2,\delta)$ with $\delta \in \Int{3,7}$ \label{fig:2d}}
 \includegraphics[width=.8\linewidth]{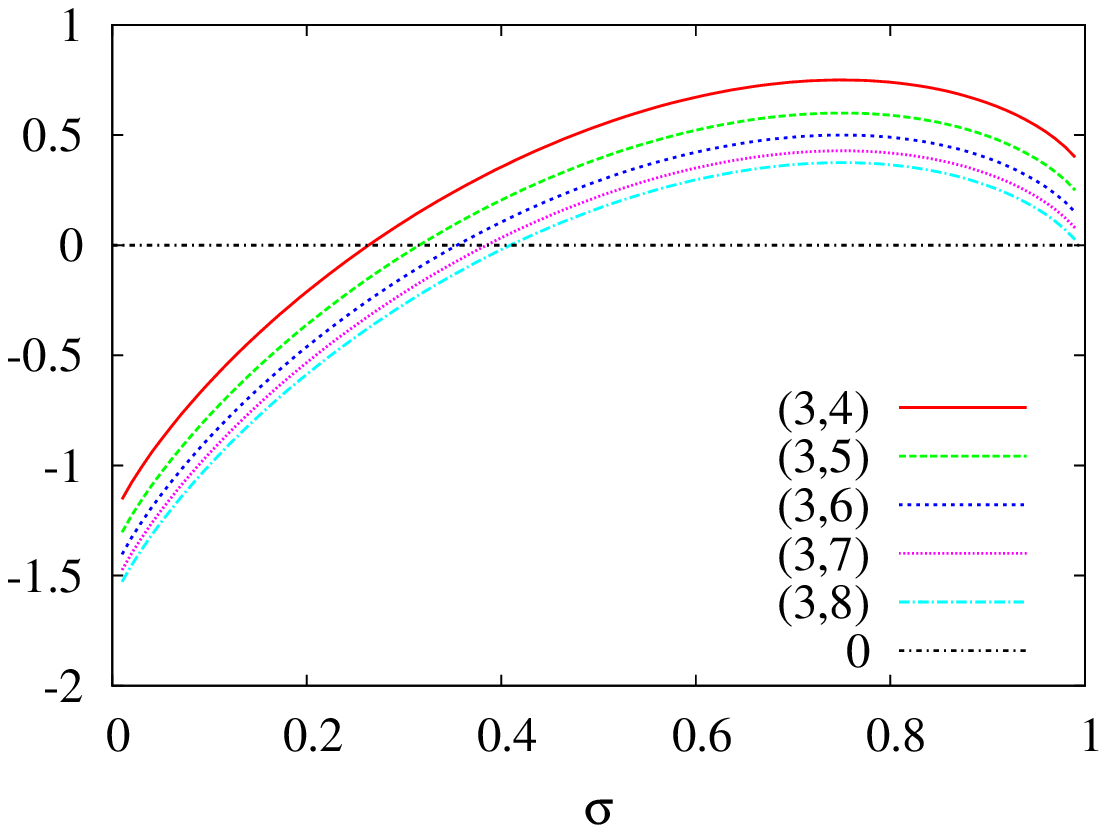}
 \caption{Growth rate $h(\sigma,0)$ for $\ens(3,\delta)$ with $\delta \in \Int{4,8}$ \label{fig:3d}}
 \end{center}
\end{figure}

Now, we plot the growth rate $h(\sigma, 0)$ for several ensembles.
Figures \ref{fig:2d} and \ref{fig:3d} plot the growth rate $h(\sigma,0)$ for $\ens(2,\delta)$ with $\delta \in \Int{3,7}$ and for $\ens(3,\delta)$ with $\delta \in \Int{4,8}$, respectively, by using \eqref{eq:gamma.5}.
Figures \ref{fig:2d} and \ref{fig:3d} show that the relative typical minimum cutsizes are strictly positive.
Moreover, the relative typical minimum cutsize monotonically increases as $\delta$ increases.
The maximum of growth rate $h(\sigma, 0)$ is achieved at $\sigma = 1/2$ and $\sigma = 3/4$ for $\ens(2,\delta)$ and $\ens(3,\delta)$, respectively.
Those agree with Lemma \ref{lem:posi}.

\subsection{Typical Minimum Cutsize \label{ssec:tmc}}
In this section, we evaluate both the right and left hand sides of \eqref{eq:para_cond2} for some ensembles by numerical examples.
In other words, the numerical examples in this section examine whether the ensemble $\ens(n,\gamma,\delta)$ satisfies the necessary condition given in Proposition \ref{prop:penc_ens}.

\begin{table}
\begin{center}
 \caption{ The left and right hand sides of \eqref{eq:para_cond2} for $\gamma=2$ \label{tab:gamma2}}
 \begin{tabular}{|c||r|r|r|r|r|r|} \hline
$\delta$ & 3 & 4 & 5 & 6 & 7 & 8  \\ \hline
$1-\gamma/\delta$ &
 0.3333 & 0.5000 & 0.6000 & 0.6667 & 0.7142 & 0.7500 \\ \hline
$\beta^{*}(0)$ &
 0.0615 & 0.1100 & 0.1461 & 0.1740 & 0.1962 & 0.2145 \\ \hline
\end{tabular}
 \caption{The left and right hand sides of \eqref{eq:para_cond2} for $\gamma=3$ \label{tab:gamma3}}
  \begin{tabular}{|c||r|r|r|r|r|r|r|} \hline
   $\delta$   & 4 & 5 & 6 & 7 & 8 & 9 \\ \hline
   $1-\gamma/\delta$ &
0.2500 & 0.4000 & 0.5000 & 0.5714 & 0.6250 & 0.6667 \\ \hline
   $\beta^{*}(0)$ & 
0.2636 & 0.3157 & 0.3545 & 0.3849 & 0.4094 & 0.4297 \\ \hline
  \end{tabular}
 \caption{The left and right hand sides of \eqref{eq:para_cond2} for $\gamma=5$ \label{tab:gamma5}} 
  \begin{tabular}{|c||r|r|r|r|r|r|r|r|} \hline
   $\delta$   & 6 & 10 & 15 & 20 & 21 & 25 \\ \hline
   $1-\gamma/\delta$ &
   0.1667 & 0.5000 & 0.6667 & 0.7500 & 0.7619 & 0.8000 \\ \hline
   $\beta^{*}(0)$  & 
   0.5570 & 0.6589 & 0.7193 & 0.7537 & 0.7589 & 0.7764  \\ \hline
  \end{tabular}
\end{center}
\end{table}

Tables \ref{tab:gamma2}, \ref{tab:gamma3} and \ref{tab:gamma5} shows the left and right hand sides of \eqref{eq:para_cond2}, i.e, $1-\gamma/\delta$ and $\beta^{*}(0)$, for $\gamma=2,3,5$, respectively.
Table \ref{tab:gamma2} shows that the ensemble $\ens(n,2,\delta)$ satisfies the necessary condition given in Proposition \ref{prop:penc_ens} for any $\delta\ge 3$.
Similarly, the ensembles $\ens(n,3,\delta)$ for $\delta \ge 5$ and the ensembles $\ens(n,5,\delta)$ for $\delta \ge 21$ satisfy the necessary condition given in Proposition \ref{prop:penc_ens} from Table \ref{tab:gamma3} and \ref{tab:gamma5}.

On the other hand, the ensembles $\ens(n,3,4)$ and $\ens(n,5,\delta)$ with $\delta \le 20$ do not satisfy the necessary condition.
In other words, we cannot parallelize the encoding algorithm by the block-diagonalization for almost all codes in those ensembles.

\section{Conclusion \label{sec:conc}}
This paper has investigated a necessary condition of parallel encodable by the block-diagonalization for the regular LDPC ensembles.
We have shown that the necessary condition depends on the minimum cutsize of the balanced hypergraph partitioning.
We have analyzed the cutsize distributions, growth rates and typical minimum cutsizes for the random hypergraphs generated by the regular LDPC ensembles.

\section*{Acknowledgment}
This work was supported by JSPS KAKENHI Grant Number 16K16007.

\bibliographystyle{IEEEtran}
\bibliography{IEEEabrv,nozaki_bib}

\end{document}